\documentclass[11pt,onecolumn]{article}
\usepackage{times}
\usepackage{amsmath}
\usepackage{amssymb}
\usepackage{xypic}
\usepackage{xspace}
\usepackage{graphicx}
\usepackage{ifpdf}
\usepackage{url,hyperref}
\usepackage{latexsym}
\usepackage{euscript}
\usepackage{xspace}
\usepackage{color}
\usepackage{makeidx}
\usepackage{picins,wrapfig}
\usepackage[boxruled,vlined,linesnumbered]{algorithm2e}
\usepackage{amsthm}

\long\def\remove#1{}

\newtheorem{theorem}{Theorem}[section] 

\newtheorem{proposition}[theorem]{Proposition}
\newtheorem{definition}[theorem]{Definition}
\newtheorem{problem}[theorem]{Problem}

\newcommand {\mm}[1] {\ifmmode{#1}\else{\mbox{\(#1\)}}\fi}

\newcommand{\ZZ}{\mathbb{Z}}

\newcommand{\G}{\widetilde{G}}

\newcommand{\B}{\mathsf{B}}
\newcommand{\C}{\mathsf{C}}
\newcommand{\Z}{\mathsf{Z}}
\newcommand{\cycle}{\gamma}

\DeclareMathOperator{\im}{im}
\DeclareMathOperator{\rank}{rank}

\DeclareMathOperator{\argmin}{argmin}
\DeclareMathOperator{\loglog}{loglog}

\newcommand{\KK}                {{\cal K}}

\newcommand{\homo}[1]   {{\sf H}_{#1}}

\newcommand{\ann}            {{\mathrm {a}}}
\newcommand{\cord}      {\alpha}
\newcommand{\matR}      {{R}}
\newcommand{\nQ}        {{\widehat{Q}}}
\newcommand{\nU}        {{\widehat{U}}}
\newcommand{\nS}        {{\widehat{S}}}

\begin{document}

\title{Annotating Simplices with a Homology Basis and Its Applications\thanks{%
Research was partially supported by the Slovenian Research Agency, 
program P1-0297 and NSF grant CCF 1064416.}}

\author{
Oleksiy Busaryev\thanks{
Department of Computer Science and Engineering,
The Ohio State University, Columbus, OH 43210, USA.
Email: {\tt busaryev@cse.ohio-state.edu}}
\quad\quad 
Sergio Cabello\thanks{
Department of Mathematics,
University of Ljubljana, Slovenia.
Email: {\tt sergio.cabello@fmf.uni-lj.si}}
\quad\quad 
Chao Chen\thanks{
Institute of Science and Technology Austria, Klosterneuburg, Austria.
Email: {\tt chao.chen@ist.ac.at}}
\quad\quad 
Tamal K. Dey\thanks{
Department of Computer Science and Engineering,
The Ohio State University, Columbus, OH 43210, USA.
Email: {\tt tamaldey@cse.ohio-state.edu}}
\quad\quad
Yusu Wang\thanks{
Department of Computer Science and Engineering,
The Ohio State University, Columbus, OH 43210, USA.
Email: {\tt yusu@cse.ohio-state.edu}}
}

\date{}
\maketitle
\begin{abstract}
\noindent Let $\KK$ be a simplicial complex and
$g$ the rank of its $p$-th homology group
$\homo{p}(\KK)$ defined with $\ZZ_2$ coefficients.
We show that we can compute a basis $H$ of $\homo{p}(\KK)$
and annotate each $p$-simplex of $\KK$ with a binary vector of length $g$
with the following property:
the annotations, summed over all $p$-simplices in any $p$-cycle $z$,
provide the coordinate vector of the homology class
$[z]$ in the basis $H$.
The basis and the annotations for all simplices can be computed
in $O(n^{\omega})$ time, where $n$ is the size of $\KK$ and
$\omega<2.376$ is a quantity so that two $n\times n$ matrices
can be multiplied in $O(n^{\omega})$ time.
The pre-computation of annotations permits answering
queries about the independence or the triviality of
$p$-cycles efficiently. \\
\\
Using annotations of edges in $2$-complexes,
we derive better algorithms for computing optimal basis and optimal
homologous cycles
in $1$-dimensional homology.
Specifically, for computing an optimal basis of $\homo{1}(\KK)$,
we improve the time complexity known for the problem from $O(n^4)$ to
$O(n^{\omega}+n^2g^{\omega-1})$. Here $n$ denotes the size of the
$2$-skeleton of $\KK$ and
$g$ the rank of $\homo{1}(\KK)$.
Computing an optimal cycle homologous to a given $1$-cycle is NP-hard
even for surfaces and an algorithm taking $2^{O(g)}n\log n$
time is known for surfaces.
We extend this algorithm to work with arbitrary $2$-complexes
in $O(n^{\omega})+2^{O(g)}n^2\log n$ time using annotations.
\end{abstract}
\section{Introduction}
Cycles play a fundamental role in summarizing
the topological information about the underlying space that
a simplicial complex represents. 
Homology groups are well known algebraic structures 
that capture topology of a space by identifying
equivalence classes of cycles.
Consequently, questions about homological characterizations of
input cycles often come up in computations dealing with
topology. For example, to compute 
a shortest basis of a homology group with a greedy approach,
one has to test several times whether the cycles in a given set
are {\em independent}. To determine the topological complexity
of a given cycle, a first level test could be deciding if it
is {\em null homologous}, or equivalently if it is
a boundary. Recently, a number of studies have been done that
concern with the computation of such topological properties of
cycles~\cite{ccl-fctpe-10,ccl-fsntc-10,cen-mcshc-09,cf-hrhl-11,dsw-alshb-10,en-mcsns-11,ew-gohhg-05,p-deeoc-10}. 

Two optimization questions about cycles have caught the attention
of researchers because of their relevance in applications:
(i) compute an optimal homology basis, which asks to compute a set of cycles
that form a basis of the corresponding homology group and whose weight
is minimum among all such basis; (ii) compute an optimal homologous cycle,
which asks to compute a cycle with minimum weight in the homology class of a given cycle. 
Chen and Freedman~\cite{cf-hrhl-11} have shown that both problems
for $p$-dimensional cycles, 
$p$-cycles in short, with $p>1$,
are NP-hard to approximate within constant factor.
Thus, it is not surprising that most of the studies have 
focused on $1$-cycles
except for a special case considered in~\cite{DHK10}. 
In this paper, we use {\em simplex annotations}
which lead to better solutions to these problems. 
We only consider homology over the field $\ZZ_2$.

\paragraph{Annotation.}
An {\em annotation} for a $p$-simplex 
is a length $g$ binary vector, where $g$ is the rank  
of the $p$-dimensional homology group. 
These annotations, when summed up for simplices
in a given cycle $z$, provide 
the coordinate vector of the homology class of $z$ in
a pre-determined homology basis. Such coordinates are only of length $g$ and
thus help us determine efficiently the topological characterization
of $z$. We provide an algorithm to compute such annotations
in $O(n^{\omega})$ time,
where $n$ is the number of input simplices and $\omega$ is a quantity
so that two $n\times n$ matrices can be multiplied in
$O(n^{\omega})$ time. It is known that 
$\omega$ is smaller than $2.376$~\cite{cw-mmap-90}.

The high-level idea for computing the annotation can be described as follows. 
We first compute an appropriate basis $Z$ of the cycle space $\Z_p(\KK)$
with the following property: 
each cycle $z\in Z$ has a \emph{sentinel}, 
which is a simplex $\sigma_z$ that appears in the cycle $z$ 
and in no other cycle from $Z$. 
We can then express any cycle $z_0$ efficiently in the basis $Z$ 
as the addition of the cycles $z$ from $Z$ 
whose sentinels $\sigma_z$ are contained in $z_0$.
Next, we compute an arbitrary homology basis $H$ of $\homo{p}(\KK)$.
The annotation of any non-sentinel simplex is simply $0$, 
while the annotation of a sentinel $\sigma_z$ is the coordinates of the 
homology class $[z]$ in the basis $H$. 
Because of linearity, the sum of the annotations
over the sentinels in a cycle gives the homology class of that cycle.
We show how matrix decomposition algorithms can be leveraged to compute these bases and annotations efficiently. 
The recent works of Milosavljevi{\'c} et al.~\cite{mms-11} and
Chen and Kerber~\cite{ck-11} also employ fast matrix operations in Computational Topology.


Annotating the simplices of a simplicial complex with elements of an algebraic
object has a long tradition in Algebraic Topology.
To work with \emph{homotopies} in a simplicial complex one can find a 
non-Abelian group called fundamental group,
described by certain relations, and assign to each 1-simplex an element 
of the group such that deciding if a path is contractible amounts 
to testing whether the product of the corresponding group elements is the identity.
This line of work dates back to Poincar\'e, but testing contractibility
is equivalent to the word problem for groups, and thus undecidable.
When the simplicial complex is a 2-manifold, 
Dehn showed in 1912 that the approach leads to efficient computation which
was further studied in~\cite{ds-99}.
We refer to Stillwell~\cite[Chapters 0, 3, 4 and 6]{stillwell} for a comprehensive treatment 
and historical account of annotations to work with homotopies.
Annotations have been used extensively to work with 1-dimensional homology in surfaces,
where they can actually be computed in linear time; 
see for example~\cite{ccl-fctpe-10,en-mcsns-11,ew-gohhg-05,p-deeoc-10}.
However, we are not aware of previous works using annotations to work with homology
in higher dimensions or in general simplicial complexes. 
For readers familiar with cohomology, it may be worth pointing out
that cocycles $\{\phi_i \}_{i=1\dots g}$ whose classes generate the 
cohomology group provide an
annotation by assigning the binary vector 
$(\phi_1(\sigma),\dots, \phi_g(\sigma))$ to simplex $\sigma$. From this
viewpoint, annotations can be seen as exposing the classical relation
between homology and cohomology groups.

\paragraph{Applications.}
Our annotation technique has the following applications.
\begin{enumerate}
\item Using the annotations for edges, we can compute an optimal
basis for $1$-dimensional homology group $\homo{1}(\KK)$
in $O(n^{\omega}+n^2g^{\omega-1})$ time, where $g$ is the first Betti number of a simplicial complex $\KK$.
This improves the previous $O(n^4)$ best known algorithm for computing an optimal homology basis 
in simplicial complexes \cite{dsw-alshb-10}. 
\item Since it is known that computing an optimal homologous cycle is NP-hard even 
for $1$-cycles~\cite{cf-qhc-07,cf-qhc-08,cen-mcshc-09} in surfaces,
Chambers et al.~\cite{cen-mcshc-09} designed an algorithm taking near-linear
time when $g$ is constant. 
Erickson and Nayyeri~\cite{en-mcsns-11} improved the running time to $2^{O(g)}n\log n$,
and Italiano et al.~\cite{insw-2011} provide an algorithm using $g^{O(g)}n\loglog n$ time.
Using our annotations together with the approach of Erickson and Nayyeri we obtain 
an algorithm for finding an optimal homologous cycle in simplicial complexes in
$O(n^{\omega})+2^{O(g)} n^2\log n$ time.
\item Using annotations for $p$-simplices, we can determine if a given $p$-cycle
is null homologous or if two $p$-cycles are homologous in time $O(tg)$ time
where $t$ is the number of $p$-simplices in the given $p$-cycles.
Given a set of $p$-cycles, we can also answer queries about
their homology independence. A set of $p$-cycles is called 
\emph{homology independent}
if they represent a set of linearly independent homology classes.
A maximal subset of homology independent cycles from a given set
of $k$ cycles with $t$ simplices
can be computed in $O(tg+(g+k)g^{\omega-1})$ time 
after computing the annotations. 
\end{enumerate} 

In many applications, $g$, the dimension of the concerned homology group 
is small and can be taken as a constant. In such cases, the applications
listed above benefit considerably, e.g., applications in 1 and 2 run 
in $O(n^{\omega})$ time. 

\paragraph{Structure of the paper.}
We introduce appropriate background in Section~\ref{sec:background}, and describe the matrix operations we use in Section~\ref{sec:vectorspace}. 
In Section~\ref{sec:annotation-2} we explain how the annotations 
for edges can be computed to work with 1-cycles.
In Section~\ref{sec:optimality} we show results on computing an optimal homology basis and an optimal homologous cycle using edge-annotations.
In Section~\ref{sec:annotation-p} we explain 
how to extend the annotation algorithm for edges to compute annotations for $p$-simplices. 
We next describe some applications of this general annotation result 
in Section ~\ref{sec:null-homology}. 
We conclude in  Section~\ref{sec:conclusions}. 

\section{Background}
\label{sec:background}
\paragraph{Homology.}
In this paper, we focus on simplicial homology over the field $\ZZ_2$;
see comments in the conclusion section for extension to other finite fields.
We briefly introduce the notations for chains, cycles, boundaries,
and homology groups of a simplicial complex, adapted to $\ZZ_2$.
The details and general case appear in any standard
book on algebraic topology such as~\cite{Munkres}.

Let $\KK$ be a simplicial complex. Henceforth, we assume
that $\KK$ is connected and 
use $\KK_p$ to denote the set of simplices in $\KK$ of dimension at most $p$.
To work with the 1-skeleton we use $V=\KK_0$, $E=\KK_1$,  
and borrow standard notation from graph theory.

A \emph{$p$-chain}
in $\KK$ is a formal sum of $p$-simplices, 
$c=\sum_{\sigma\in \KK_p}\alpha_\sigma\sigma$, $\alpha_\sigma\in \ZZ_2$.
The set of $p$-chains forms a vector space 
$\C_p(\KK)$ under $\ZZ_2$-addition where the empty chain
plays the role of identity $0$. 
%
The chain group $\C_p$ is in one-to-one 
correspondence to the family of subsets of $\KK_p$. Hence
$\C_p$ is isomorphic to the $n_p$-dimensional binary vector space 
$(\ZZ_2)^{n_p}$, where $n_p$ 
is the number of $p$-simplices in $\KK$. 
A natural basis of $\C_p$ consists of the $p$-simplices in $\KK$.
In this basis, the coordinate vector of a $p$-chain is the incidence vector
telling which $p$-simplices appear in the corresponding subset.

The boundary of a $p$-simplex is a $(p-1)$-chain consisting of the set of 
its $(p-1)$-faces. This can be linearly extended to a \emph{boundary map}
$\partial_p\colon\C_{p}\rightarrow\C_{p-1}$, where the boundary of a chain
is defined as the sum of the boundaries of its elements. 
Using the natural bases of $\C_{p}$ and $\C_{p-1}$, 
computing the boundary of a $p$-chain corresponds to 
multiplying the chain vector
with a boundary matrix $[b_1\,\,b_2\cdots b_{n_p}]$
whose column vectors are boundaries
of $p$-simplices. We slightly abuse the notation and denote 
the {\em boundary matrix} also with $\partial_p$.

We define the \emph{group of $p$-cycles}
as the kernel of $\partial_p$, $\Z_p := \ker \partial_p$, and define 
the \emph{group of $p$-boundaries} as the image of $\partial_{p+1}$, 
$\B_p := \im \partial_{p+1}$. The latter is a subgroup of the former.
The \emph{$p$-th homology group} $\homo{p}$ is the quotient $\Z_p/ \B_p$.
Each element in $\homo{p}$, called a \emph{homology class}, 
is an equivalence class of $p$-cycles whose
differences are $p$-boundaries. For any $p$-cycle $z$, we use $[z]$ to denote 
the corresponding homology class. 
Two cycles are \emph{homologous} when they belong 
to the same homology class.
Note that $\Z_p$, $\B_p$, and $\homo{p}$ are also vector spaces.
We call their bases a \emph{cycle basis},
a \emph{boundary basis}, and a \emph{homology basis} respectively.
The dimension of the $p$-th homology group 
is called the \emph{$p$-th Betti number}. We will denote it by $g$.
A set of $p$-cycles $\{z_1,\cdots,z_{g}\}$ 
is a \emph{homology cycle basis} if the set of classes 
$\{[z_1],\cdots,[z_{g}]\}$ forms a homology basis.

\paragraph{Optimization problems.} 
Given a simplicial complex $\KK$, 
exponentially many cycles may belong to 
a homology class $[z]$.
We consider an optimization problem over such a set with
all $p$-cycles assigned well-defined weights. Given a non-negative 
real weight $w(\sigma)$ for each $p$-simplex $\sigma$, 
we define the weight of a cycle as the sum of the weights of its simplices,
$w(z)=\sum_{\sigma\in z} w(\sigma)$. 
For example, when $p=1$ and the weights are the lengths of the edges,
the weight of a cycle is its length and
the optimization problem seeks for the shortest 
cycle in a given class. Formally, we state:
\begin{problem}
Given a simplicial complex and a cycle $z$, find the cycle
$\argmin_{z_0\in [z]}w(z_0).$
\end{problem}
Next, we consider an optimization problem over the set of all 
homology cycle bases.
The weight of a homology cycle basis $H$ is defined as the sum of the 
weights of its elements,
$w(H)=\sum_{z\in H} w(z)$. Note that a simplex may contribute to the weight 
multiple times if it belongs to multiple cycles in the basis $H$. 
Formally, we have the following problem.
\begin{problem}
\label{prob:OBAS}
 Given a simplicial complex, find a homology cycle basis $H$ with minimal $w(H)$.
\end{problem}

\section{Efficient Matrix Operations}
\label{sec:vectorspace}

Under $\ZZ_2$ coefficients, the groups $\C_p$, $\Z_p$, $\B_p$, and $\homo{p}$ are all vector spaces.
Linear maps among such spaces or change of bases within the same
space can be represented by matrices and operations on them.
Our algorithm computes simplex annotations via manipulations of 
such matrices and bases. 
Several of our computations use the following concept.


\begin{definition}[Earliest Basis]
Given a matrix $A$ with rank $r$, the set of columns 
$B_{opt}=\{a_{i_1},\cdots,a_{i_r}\}$ is called the
\emph{earliest basis} if the
column indices $\{i_1,\cdots,i_r\}$ 
are the lexicographically smallest index 
set such that the corresponding columns of $A$ have full rank.
\end{definition}
For convenience, 
we often use the same symbol to denote both a set of vectors 
and the matrix they form and 
denote by $B_{opt}$ also the matrix $[a_{i_1}\,\,a_{i_2}\cdots a_{i_r}]$.
It is convenient to consider the following alternative view of the earliest basis:
a column vector of $A$ is in the earliest basis if and only if it does not
belong to the subspace generated by column vectors to its left. 

We next summarize the operations on matrices that we need. 
For simplicity, we assume that the matrix multiplication exponent $\omega>2$; 
otherwise, some additional logarithmic terms appear in the running times.
\begin{proposition}\label{lem:linear}
	Let $A$ be an $m\times n$ matrix of rank $r$ 
	with entries over $\ZZ_2$ where $m\le n$. 
	\begin{itemize}
		\item[\textup(a\textup)] If $A$ is square and has full rank, one can compute its inverse $A^{-1}$ in $O(n^\omega)$ time. 
		\item[\textup(b\textup)] There is an $O(n^\omega)$ time algorithm to compute the earliest basis $B_{opt}$ of $A$. %
		\item[\textup(c\textup)] In $O(n^\omega)$ time, one can compute the coordinates of all columns of $A$ in the earliest basis $B_{opt}$. Formally, one can compute $AP = B_{opt} [I_r \mid \matR]$, where $P$ is a permutation matrix, 
		$I_r$ is an $r\times r$ identity matrix, and $\matR$ is an $r\times (n-r)$ matrix.
	\end{itemize}
\end{proposition}
\begin{proof}
	Item (a) appears in Bunch and Hopcroft~\cite{bc-tfifm-74}; alternatively, see \cite{introduction_algorithms}.
	Item (b) and (c) follows from the \emph{LSP-decomposition}~\cite{ibarra1982generalization,jeannerod2006lsp},
	which can be computed in $O(n^\omega)$ time. 
	We restate the following result for a column version,
where we transpose both sides of the standard LSP-decomposition 
and rename the matrices.

\smallskip
\noindent{\bf LSP-decomposition} \cite{jeannerod2006lsp}. {\it Given an $m\times n$ matrix $A$, one can compute in $O(n^\omega)$ time a decomposition
$A=QSU$,
where $Q$ is an $m\times m$ permutation matrix, 
$U$ is an $n\times n$ upper unitriangular matrix, 
$S$ is an $m\times n$ matrix with $r = \rank(A)$ non-zero columns 
which are linearly independent. Furthermore, $S$ is lower unitriangular 
if we remove all zero columns and the lowest $m-r$ rows.} 

\smallskip
Notice that the permutation matrix $Q$ only permutes rows and thus does not
affect the computation of the earliest basis of the column rank.
By definition, the indices of non-zero columns of $S$ (called the \emph{column rank profile})
are the indices of the earliest basis, and thus can be computed by 
applying LSP-decomposition once. Therefore item (b) follows.

Next, we prove item (c).
Due to item (b), we know the indices 
of the earliest basis $B_{opt}$ in $A$
and we can compute 
a column permutation matrix $P$ so that $AP$ 
has the first $r$ columns as this earliest basis. 
Next, we compute the LSP-decomposition of the matrix $AP = \nQ \nS \nU$. 
Since indices of the non-zero columns in $\nS$ correspond to those of the earliest basis of $AP$, the last $n-r$ columns of $\nS$ are necessarily zero. 
Hence we can rewrite the LSP decomposition as 
$$
AP ~~ = ~~\nQ\nS\nU ~~ = ~~ \nQ[\nS_1 \mid 0]\left[\begin{array}{cc}\nU_{11} & \nU_{12}\\0 & \nU_{22}\end{array}\right]
$$
where $\nS_1$ has size $m\times r$, $\nU_{11}$ has size 
$r\times r$ and is upper unitriangular, 
$\nU_{12}$ has size $r\times (n-r)$, and $\nU_{22}$ has 
size $(n-r)\times (n-r)$.
Evaluating the right side, we have 
$$
AP ~~ = ~~ \nQ\nS_1\nU_{11} [I_r \mid \nU_{11}^{-1}\nU_{12}].
$$ 
Since the first $r$ columns of the matrix on right hand side equal to
$\nQ\nS_1\nU_{11}$, and $B_{opt}$ consists of the first $r$ columns of 
$AP$ by definition of $P$, 
we see that $\nQ\nS_1\nU_{11}=B_{opt}$. Setting $R:=\nU_{11}^{-1}\nU_{12}$
we obtain the desired decomposition $AP = B_{opt}[I_r\mid \matR]$. 
The algorithm involves one computation of an earliest basis, 
one LSP-decomposition, and constant number of 
matrix inversions and multiplications. Since each of these operations takes 
$O(n^\omega)$ time, item (c) follows.
\end{proof}

\section{Annotating Edges}
\label{sec:annotation-2}

Let $\KK$ be a given simplicial complex.
First, we define annotations in general terms using $g$ for the dimension of $\homo{p}(\KK)$.

\begin{definition}[Annotations]
 An \emph{annotation} for $p$-simplices is a function $\ann\colon \KK_p \rightarrow (\ZZ_2)^g$ with the following property:
 any two $p$-cycles $z$ and $z'$ belong to the same homology class if and only 
 $$
	\sum_{\sigma\in z} \ann(\sigma) ~=~ \sum_{\sigma\in z'} \ann(\sigma).
 $$
 Given an annotation $\ann$, the \emph{annotation} of any $p$-cycle $z$ is defined by $\ann(z)=\sum_{\sigma\in z} \ann(\sigma)$.
\end{definition}

We will construct annotations using coordinate
vectors of cycles in a homology basis.
Let $H=(h_1,h_2,\dots, h_g)$ be a basis of the vector space $\homo{p}(\KK)$.
For a $p$-cycle $z$, if $[z]=\sum_{i=1}^g \lambda_i h_i$
where each $\lambda_i\in \ZZ_2$, then the coordinate vector of $[z]$ in $H$ is
$(\lambda_1,\dots,\lambda_g)\in (\ZZ_2)^g$. The question is how to
annotate the $p$-simplices so that the sum of annotations in the simplices of $z$ 
gives $(\lambda_1,\dots,\lambda_g)$.

In this section, we explain the technique for annotating edges. An extension
to $p$-simplices is explained in Section \ref{sec:annotation-p}.
We compute edge annotations in three steps. 
First, we construct a cycle basis $Z$ in which 
any cycle can be expressed in simple and efficient terms.
Second, we find a homology cycle basis $H$. Last, we compute the homology of 
each cycle in $Z$ in the homology cycle basis $H$.
From this information, one can compute the homology class of any other 
cycle using vector sums in the coordinate
system provided by $H$. The approach is based on using a spanning tree of the
1-skeleton to generate the space of cycles. The approach of using a spanning
tree to generate the fundamental group and then obtain the homology group,
is well known in topology and has been used extensively;
see for example~\cite{e-dgteg-03,ew-gohhg-05,stillwell}

\paragraph{Step 1: Computing a cycle basis $\boldsymbol{Z}$.}
Let us fix throughout this section a spanning tree $T$ 
in the $1$-skeleton of $\KK$; it contains $n_0-1$ edges.
Let $k=n_1-n_0+1$ be the number of edges in $E\setminus E(T)$.
We fix an enumeration $e_1,\dots, e_{n_1}$ of the edges 
of $E$ with the property
that the edges $e_1,\dots ,e_k$ are precisely the edges of $E\setminus E(T)$.
Thus, $e_{k+1},\dots, e_{n_1}$ are the edges of $T$.
The edges of $E\setminus E(T)$ are called \emph{sentinel edges}, while
the edges of ${E(T)}$ are \emph{non-sentinel edges}.

For any sentinel edge $e\in E\setminus E(T)$, 
denote by $\cycle(T,e)$ the 
cycle corresponding to the unique simple path that connects the endpoints 
of $e$ in $T$ plus the edge $e$. We call it a \emph{sentinel cycle}.
Let $Z$ be the set of such sentinel cycles 
$\{\cycle(T,e_1),\cycle(T,e_2), \dots, \cycle(T, e_k)\}$.
We have the following property: a sentinel edge $e_i$ belongs
to a sentinel cycle $\gamma(T,e_j)$ if and only if $i=j$.
For completeness, we set $\cycle(T,e)=0$ when $e$ belongs to $T$.
The following result is probably folklore.

\begin{proposition}[Cycle basis]\label{lem:cyclebasis}
	$Z$ is a cycle basis and for any cycle $z\in \Z_1$ we have 
		$z = \sum_{e\in z} \cycle(T,e)$.
\end{proposition}
	\begin{proof}
		Since the edge $e_1\in E\setminus E(T)$ does not appear in 
	$\cycle(T,e_2),\dots, \cycle(T,e_k)$, the cycle $\cycle(T,e_1)$ is 
	linearly independent of $\cycle(T,e_2),\dots, \cycle(T,e_k)$. 
	The same argument applies to any cycle $\cycle(T,e_i)\in Z$,
		and thus the cycles of $Z$ are linearly independent. 
	Since $k = \dim(\Z_1)$, $Z$ is indeed a basis for $\Z_1$. 

		To show the second  half of the claim, fix an arbitrary cycle $z$. 
	Fix a vertex $s$ of the tree $T$ and for any vertex $u$ of $T$, 
	let $T[s,u]$ denote the unique simple path in $T$ connecting $u$ to $s$. 
	We then have:
	$$
	z ~=~ \sum_{e = uv \in z} e ~=~ \sum_{e = uv \in z} (T[s,u] + e + T[s,v]) ~=~ 
			\sum_{e\in z} \cycle(T,e).
	$$
	The second equality holds as every vertex in the cycle $z$ is the endpoint of an even number of edges of $z$. 
	\end{proof}

\paragraph{Step 2: Computing a homology cycle basis $\boldsymbol{H}$.}
In this step, we compute a homology cycle basis $H$ from $Z$
with the help of Proposition \ref{lem:linear}(b). 
Specifically, we construct a new matrix $[\partial_2\mid Z]$ with
the submatrix $Z$ being formed by the chain vectors of cycles in $Z$.
We compute the 
earliest basis $\widetilde{Z} = [B \mid H]$ of $[\partial_2\mid Z]$ 
where $B$ contains the first 
$r = \rank(\partial_2)$ columns of $\widetilde{Z}$.
Since the set of columns of $\partial_2$ generates the boundary group, 
by the definition of earliest basis, it is necessary that the columns 
in $B$ come from $\partial_2$ and form a boundary cycle basis. 
Since $Z$ and hence $\partial_2 \cup Z$
generates the cycle group,
the remaining columns of $\widetilde{Z}$, namely $H$, 
form a homology cycle basis.

\paragraph{Step 3: Computing annotations.}
Finally, for elements of $Z$ we compute their coordinates in 
the cycle basis $\widetilde{Z}$. For each 
sentinel cycle $z=\cycle(T,e)$, we compute its 
coordinate vector in $\widetilde{Z}$ by solving the linear system
$\widetilde{Z} x=z$. The last $g$ entries
of $x$ give its coordinates in the basis $H$.
We use this length $g$ vector as the annotation of the
sentinel edge $e$. 
We can compute annotations for all sentinel edges together by 
solving $\widetilde{Z}X=Z$ and taking the last $g$ rows of the solution $X$.
For a non-sentinel edge, we simply set its annotation to be the zero vector.
\parpic[r]{\includegraphics[scale=.55]{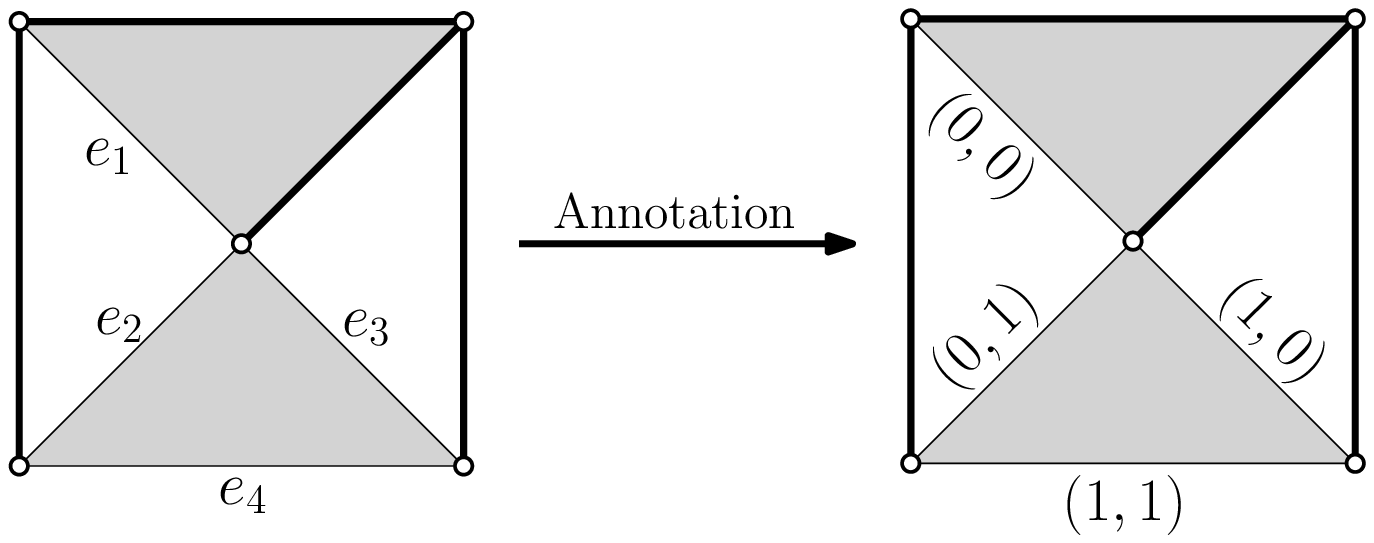}}
An example of annotation for a $2$-complex is shown on right. 
The edges of a spanning tree are shown with thicker edges. The edges $e_1,e_2,e_3,e_4$ are sentinel edges. The cycles given by sentinel edges $e_2$ and $e_3$ form the homology cycle basis $H$ computed by the algorithm. We show the annotations for the sentinel edges; all other edges get annotation $(0,0)$.
The annotation of $(1,1)$ for $e_4$ makes it possible to evaluate
the cycle $e_2e_3e_4$ to $(0,0)$ as it is null-homologous and also evaluate the
outer boundary to $(1,1)$ as it is homologous to the sum of the two holes.

\begin{theorem}
The algorithm described above computes an annotation of length
$\dim(\homo{1}(\KK))$ for the edges of a $2$-complex $\KK$ in $O(n^\omega)$ time,
where $n$ is the size of $\KK$.
\label{thm:edgeannotations}
\end{theorem}
\begin{proof}
From Step 3, the annotation of a sentinel edge $e$ is exactly the coordinate vector of the homology class $[\cycle(T,e)]$.
It then follows from Proposition \ref{lem:cyclebasis} that, for any cycle $z$, 
the coordinate vector of the homology class $[z]$ is simply the summation of annotations of all edges in $z$. 
For the time complexity, notice that Step 1 requires computing a spanning tree and the cycle basis $Z$, which takes $O(n^2)$ time.
Steps 2 and 3 take $O(n^\omega)$ time because of Proposition \ref{lem:linear}(b) and (a) respectively.
\end{proof}

\section{Optimality for \texorpdfstring{$\boldsymbol{1}$}{1}-cycles}
\label{sec:optimality}
\subsection{Shortest homology basis}
\label{sec:shortest}

In this section we discuss the problem of computing an optimal homology basis for 
one dimensional homology $\homo{1}$. 
The optimal homology cycle basis here is the \emph{shortest homology basis}
since we minimize the weights / lengths.
We present an efficient algorithm that combines the 
approach of Erickson and Whittlesey~\cite{ew-gohhg-05}
and our annotation technique.
The approach restricts the search to a well-structured family of cycles,
represents each cycle in this family with a length-$g$ binary vector, 
and then reduces the computation to the problem of finding an 
earliest basis in a matrix of size $g\times n^2$.

For each vertex $s\in V$, let $T_s$ be the {\em shortest path tree} from $s$ with respect to
the weight function.
Denote by $Z_s$ the set of sentinel cycles corresponding to this tree $T_s$ and $\Pi$ the union
of $Z_s$ for all $s\in V$, that is, 
$$\Pi ~=~ \bigcup_{s\in V}Z_s 
	  ~=~ \bigcup_{s\in V} \{ \cycle(T_s,e)\mid e\in E\setminus E(T_s)\}.
$$
The following property was noted by 
Erickson and Whittlesey~\cite{ew-gohhg-05}.
See also Dey, Sun, and Wang~\cite{dsw-alshb-10} for an extension.
\begin{proposition} \label{lem:characterization}
	If we sort the cycles of $\Pi$ in non-decreasing order of their 
weights, the earliest basis of $\Pi$ is a shortest homology basis. 
\end{proposition}

\begin{theorem}
	Let $\KK$ be a simplicial complex of size $n$.
	We can find a shortest homology basis in 
time $O(n^{\omega}+ n^2g^{\omega-1})$ where $g = \rank(\homo{1}({\KK}))$.
\label{sbasis-thm}
\end{theorem}
\begin{proof}
By Theorem \ref{thm:edgeannotations} we compute annotations for 
all edges in $O(n^{\omega})$ time. Let $\ann(e)$ 
denote such annotation for any edge $e$, and $\ann(z)=\sum_{e\in z}\ann(e)$ 
for any $1$-cycle $z$.

Next we compute annotations for all cycles $z\in \Pi$. Instead of computing
them one by one, we annotate all cycles in $Z_s$ at once for each $s$. 
Given a fixed $s$, we first compute $T_s$ in $O(n\log n)$ time. 
We assign a $g$-long label $\ell(x)$  
to each vertex $x\in V$.
The label $\ell(x)$ is the label $\ell(x')$
of its parent $x'$ plus the annotation of the edge $xx'$, $\ann(xx')$.
We compute labels for all vertices in $O(ng)$ time by a 
breadth-first traversal of $T_s$. Afterward, the annotation of 
any sentinel cycle $\cycle(T_s,xy)\in Z_s$ is computed 
in $O(g)$ time as $\ell(x)+\ell(y)+ \ann(xy)$.
Thus, we can compute the annotations for all cycles in 
$Z_s$ in $O(ng)$ time given $T_s$ and edge annotations.
To annotate all cycles of $\Pi$, 
we repeat the procedure for all source vertices $s$. 
Computing annotations for all cycles thus takes 
$O(n^2g + n^2 \log n)$ time.

\remove{
	We first fix a spanning tree $T$ in $G$ and compute 
a homology basis $H$ together with the labels $\cord_{H}([\cycle(T,e)])$ 
using Proposition~\ref{edge-annotate}.  This takes $O(n^{\omega})$ time.
	
	Next, we compute the homology type $h(z)=\cord_{H}([z])$ 
for all cycles $z$ of $\Pi$. Consider a fixed vertex $s$ of $V$; 
we are to compute the homology types for the cycles 
$\{ \cycle(T_s,e)\mid e\in E\setminus E(T_s)\}$. To each vertex $x$ of $V$ 
we assign a $g$-long label $\ell(x)$
which is the sum of the labels $\cord_{H}([\cycle(T,e)])$ 
along the edges in the path from $s$ to $x$ in $T_s$. 
Using a breadth-first traversal of $T_s$ and using that the 
label $\ell(x)$ of a vertex $x$ is the level $\ell(x')$
of its parent $x'$ plus $\cord_{H}([\cycle(T,xx')])$, 
we can compute all such labels in $O(ng)$ time.
The homology type of any cycle $\cycle(T_s,xy)$ is computed 
in $O(g)$ time as $\ell(x)+\ell(y)+ \cord_{H}([\cycle(T,xy)])$.
Thus, we can compute the labels for all cycles in $\Pi$ defined by $T_s$ 
in $O(ng)$ time.
To compute the labels for all cycles of $\Pi$, 
we repeat the procedure for all source vertices $s$.
This takes $O(n^2g)= O(n^2g^{\omega-1})$ time.
}

Since annotations of cycles give us the homology classes they belong to,
we can use them to find a shortest homology basis.
We sort cycles in $\Pi$ in non-decreasing order of their weights 
in $O(n^2\log n)$ time.
Let $z_1,z_2,z_3,\dots$ be the resulting ordering. 
We construct a matrix $A$ whose $i$th column is the 
vector $\ann(z_i)$, 
and compute its earliest basis. 
By Proposition~\ref{lem:characterization}, the cycles defining 
the earliest basis of $A$ form 
a shortest homology basis.
Since there are up to $n^2$ elements in $\Pi$, 
the matrix $A$ has size $g\times n^2$,
and thus it is inefficient to compute its earliest basis
using Proposition~\ref{lem:linear} directly. 
Instead, we use the following iterative method
to compute the 
set $J$ of indices of columns that define the earliest basis.

We partition $A$ from left to right into submatrices $A=[A_1|A_2|\cdots]$
where each submatrix $A_i$ 
contains $g$ columns with the possible exception of the last submatrix
which contains at most $g$ columns.
Initially, we set $J$ to be the empty set.   
We then iterate over the submatrices $A_i$ by increasing index. 
At each iteration we compute the earliest basis for 
the matrix $[A_{J} | A_i]$ where $A_{J}$ is the
submatrix whose column indices are in $J$. We then set $J$ to 
be the indices from the resulting earliest basis, 
increment $i$, and proceed to the next iteration.
At each iteration we need to compute the earliest basis in a 
matrix with $g$ rows and at most $|J|+g \le 2g$ columns.
There are at most $O(n^2/g)$ iterations each taking
$O(g^\omega)$ time. 

We obtain the claimed time bound by adding up the time to annotate edges, 
annotate cycles in $\Pi$,
and compute the earliest basis.
\end{proof}

\subsection{Shortest homologous cycle}
\label{sec:localization}

In this section, we show how to compute the shortest cycle in a 
given one-dimensional homology class.
In fact, within the same running time, we can compute a 
shortest cycle in each homology class.
The idea is to use covering graphs, and it
closely resembles the approach of 
Erickson and Nayyeri~\cite{en-mcsns-11}. 
We skip most of the details because of this similarity. 
Nevertheless, our main contribution is the use of the 
annotations from Section~\ref{sec:annotation-2}.

\parpic[r]{\includegraphics[scale=.65]{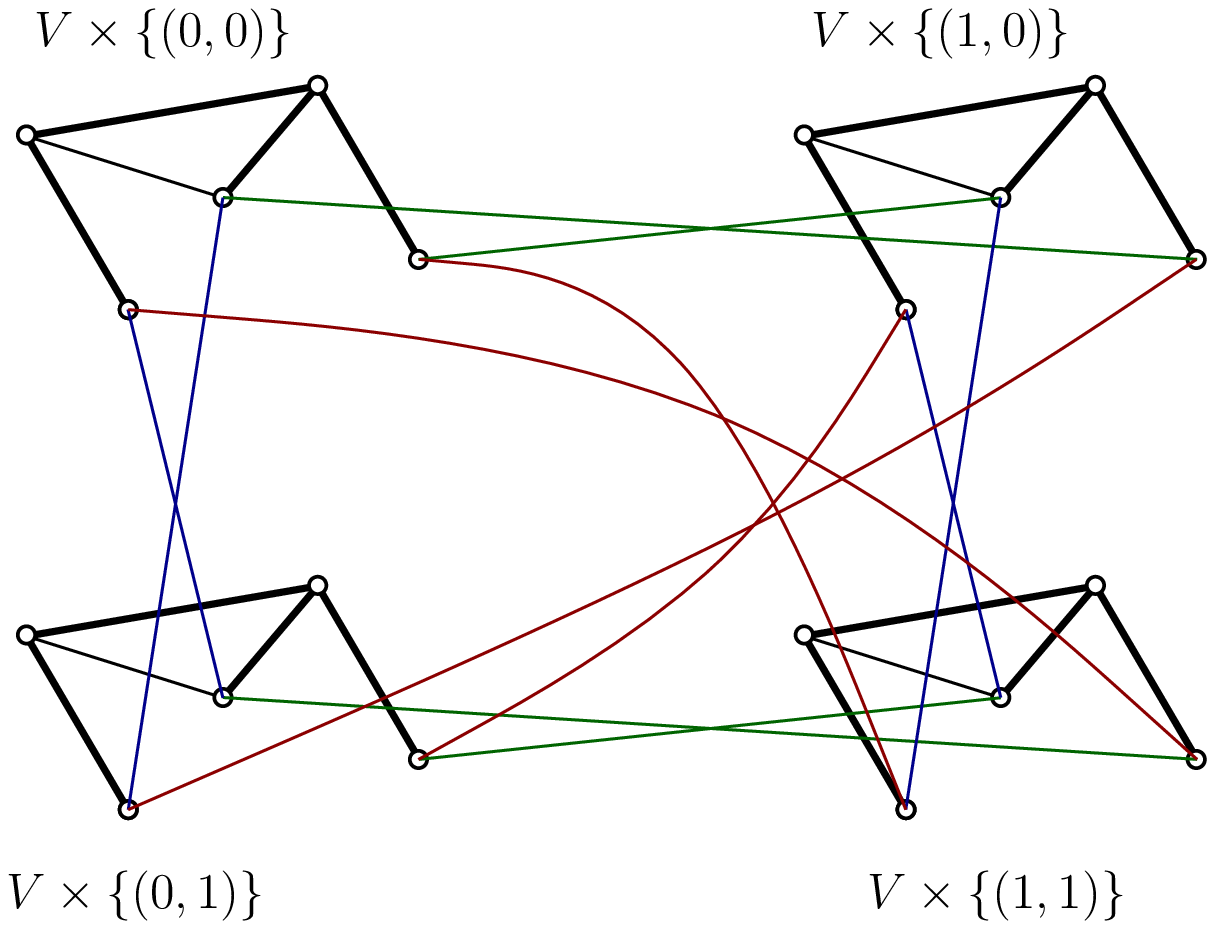}}
Let $G=(V,E)$ be the 1-skeleton of $\KK$.
We first compute an annotation $\ann\colon E\rightarrow (\ZZ_2)^g$, 
as given by Theorem~\ref{thm:edgeannotations}.
A walk in $G$ is a sequence of vertices $x_0 x_1 \ldots x_t$ connected
by edges in $E$. It is closed if $x_0=x_t$.
In this section we keep using the term cycle for elements of $\Z_1$.
Each closed walk in $G$ defines a cycle, 
where only edges appearing an odd number of times in the walk are kept. 
The annotation $\ann(w)$ of a walk $w=x_0x_1 \ldots x_t$ is defined 
as the sum of the annotations of its edges $x_{i-1} x_i$
for $i=1,\dots ,t$.
Notice that the annotation $\ann(w)$ of a closed walk $w$ is
the annotation of the cycle defined by $w$, as annotations in
edges appearing an even number of times in the walk cancel out.

We construct a covering graph $\G$ of the 1-skeleton of $\KK$, 
defined as follows:
\begin{itemize}
	\item $V(\G)= V\times (\ZZ_2)^g$.
	\item vertex $(v,h)\in V\times (\ZZ_2)^g$ is 
		adjacent to $(v',h')\in V\times (\ZZ_2)^g$
		if and only if $e=vv'$ is an edge of $E$ and 
		$h'=h+\ann(e)$.
		The weight of an edge $(v,h)(v',h')$ is the weight of $vv'$.
\end{itemize}

The graph $\G$ has $n_0\cdot 2^g$ vertices and $n_1\cdot 2^g$ edges. 
The covering graph for the example shown previously 
in section~\ref{sec:annotation-2} for annotation is
depicted above. 
The second coordinate of a vertex $(v,h)\in V(\G)$ is used to encode
the homology of cycles, as we will see.
Similar covering graphs are used, for example, in~\cite{ccl-fctpe-10,en-mcsns-11,p-deeoc-10}.

\begin{proposition}
	For all $h\in (\ZZ_2)^g$,
	we can compute
	a shortest walk $w_h$ in $G$ among all closed walks with annotation $h$
	in $O(2^g n^2(g+\log n))$ time.
\label{lem:loc-one-comp}
\end{proposition}
\remove{
\begin{proof}
We first argue that there is a bijection between closed walks in $G$ through a 
vertex $v$ with annotation $h$ and walks in $\G$ 
from vertex $(v,0)$ to vertex $(v,h)$.
Indeed, assume first that there is a closed walk 
$w=x_0x_1\cdots x_t$ in $G$ with $x_t=x_0=v$.
By the definition of annotation we have
$$
\ann(w) ~=~ \sum_{i=1,\dots, t} \ann(x_{i-1} x_{i}).
$$
Define $y_0=(x_0,0)$ and define, for each vertex $x_i$ in $w$, the vertex
$y_i=(x_i,\sum_{j\le i}\ann(x_{j-1}x_j))$.
By construction, there is an edge in $\G$ between $y_i$ and $y_{i+1}$,
and therefore $y_0 y_1\cdots y_t$ is a path in $\G$ from $(v,0)$ to $(v,\ann(w))$.
Conversely, for any walk in $\G$ from vertex $(v,0)$ to vertex $(v,h)$,
the projection into the first coordinate provides a closed
walk through $v$ whose annotation is $h$.
	
For all $v\in V$, we compute the shortest paths from $(v,0)$ to all vertices in $\G$, 
and record the lengths.
This is $|V|$ computations of shortest path trees in $\G$ and 
thus takes $O(n (2^g n \log (2^g n))= O(2^g n^2 (g+\log n))$ time
(Erickson and Nayyeri~\cite{en-mcsns-11} can speed up this step because in their case
$\G$ is embedded in a surface.)
The closed walk $w_h$ is then obtained by 
considering the shortest path in $\G$ between $(v,0)$ and $(v,h)$, 
over all vertices $v\in V$, and then taking its
projection onto $G$.
\end{proof}
}

We say that a cycle is \emph{elementary} if 
it is connected and each vertex is adjacent to at most two edges of the cycle.
Each cycle is the union of edge-disjoint elementary cycles. 
First, we bound the number of elementary cycles in optimal solutions and then use
dynamic programming across annotations and the number of elementary cycles to obtain the following

\begin{proposition}\label{prop:fewcycles}
The shortest cycle in any given homology class consists of at most 
$g$ elementary cycles.
\end{proposition}
\remove{
\begin{proof}
We can assume that the homology class is nonzero, 
as otherwise the result is trivial.
Assume for contradiction that the 
shortest cycle $z$ in a given homology class contains more than $g$ 
elementary cycles $z_1,\cdots,z_t$, where $t>g$. 
Each elementary cycle is a cycle by itself. 
Since the homology classes of 
these cycles are not linearly independent, there exists a set
of indices $I\subset \{1,\dots, t\}$ such that $0=\sum_{i\in I} [z_i]$.
Notice that $I\not=\{1,\dots, t\}$ as otherwise 
$[z]=\sum_{i\in \{1,\dots, t\}} [z_i] = 0$. The cycle 
$
	z' = \sum_{i\in \{1,\dots, t\}\setminus I} z_i 
$ 
is strictly shorter than $z$, and represents the same homology 
class as $z$. This contradicts the assumption
that $z$ is shortest in its class.
\end{proof}
}


\begin{theorem}\label{th:OHCP}
	In $O(n^\omega + 4^g g + 2^g n^2(g+\log n))= 
O(n^\omega) + 2^{O(g)} n^2\log n$ time 
	we can compute the shortest homology cycle for all homology
classes in $\homo{1}$.
\end{theorem}
\remove{
\begin{proof}
	We compute annotations for the edges in $O(n^\omega)$ time
	using Theorem~\ref{thm:edgeannotations}.
	For each $h\in (\ZZ_2)^g$ we compute the closed walk $w_h$ given 
	in Proposition~\ref{lem:loc-one-comp}. This takes $O(2^g n^2(g+\log n))$ time.
	  
	For any $h\in (\ZZ_2)^g$ and any integer $k\in [1,g]$, 
	we define $C(h,k)$ as follows
	$$
		C(h,k) = \begin{cases}
			\mbox{length of $w_h$}  &	\mbox{if $k=1$;}\\
			\min \{C(h_1,k-1)+C(h_2,1) \mid h=h_1+h_2\}&	\mbox{if $k>1$.}
		\end{cases}
	$$
	It is straightforward to see by induction
	that $C(h,k)$ is an upper bound 
	on the length of the shortest cycle with annotation $h$. 
	Most interestingly, $C(h,g)$ is the length of the shortest cycle
	with annotation $h$. 
	Indeed, if the shortest cycle with annotation $h$
	consists of the elementary cycles $z_1,\dots, z_t$ with $t\le g$,
	then it follows by induction that, for any $I\subset \{1,\dots t\}$, 
	$C \left( \sum_{i\in I} \ann(z_i), |I| \right)$ is the length of $\left(\sum_{i\in I} z_i \right)$.
	
	A standard dynamic programming algorithm to compute $C(h,k)$ takes $O(2^g)$ time per element,
	for a total of $O(4^gg)$ time. We then return for each $h\in (\ZZ_2)^g$ the cycle
	defining $C(h,g)$. (A shortest cycle homologous to a given cycle $z$
	is obtained as the cycle defining $C(\ann(z),g)$.) 
\end{proof}
}

\section{Annotating \texorpdfstring{$\boldsymbol{p}$}{p}-simplices}
\label{sec:annotation-p}

In this section, we show how to compute annotations for $p$-simplices. 
Notice that the only thing we need to generalize is the first step:
find a set $\Sigma$ of $p$-simplices (\emph{sentinel simplices}) with cardinality $\dim{(\Z_p)}$
and a cycle basis $Z=\{ z_\sigma \}_{\sigma\in \Sigma}$ (\emph{sentinel cycles}) for the $p$-cycle group $\Z_p$
with the property that $z_\sigma$ contains $\sigma'\in \Sigma$ if and only if $\sigma=\sigma'$.
With this property, any $p$-cycle $z$ can then be written as
$z=\sum_{\sigma\in z\cap \Sigma} z_\sigma.$ Taking $z_\sigma=0$ for
all $\sigma\not\in \Sigma$, we have
$z=\sum_{\sigma\in z}z_\sigma$. 
With such a basis, we proceed with Step 2 and 3 just like in the case for edges to annotate $p$-simplices. 
Below, we explain how to compute such a cycle basis $Z$.

In the case for annotating edges, we first fix a spanning tree. 
The boundaries of its edges form a $0$-dimensional boundary 
basis. Any of the remaining edges when added to the tree
creates a unique $1$-cycle which is kept associated with this edge as a sentinel cycle.
For $p$-simplices, $p>1$, we do not have a spanning tree, but
Proposition \ref{lem:linear}(c) provides us an algebraic tool that 
serves the same purpose.

Specifically, consider the $n_{p-1} \times n_p$ boundary
matrix $\partial_p$ of rank $r$, where the $i$-th column in 
$\partial_p$ corresponds to the ($p-1$)-boundary of $p$-simplex $\sigma_i$. 
Using Proposition \ref{lem:linear}(c)
we can obtain an $n_p \times n_p$ matrix $P$,
an $n_{p-1} \times r$ matrix $B_{opt}$, and an $r\times(n_p-r)$ matrix $\matR$ so that
$$
	\partial_p P = B_{opt}[I_r\mid \matR]. 
$$
Notice that $P$ permutes the 
$p$-simplices so that the first $r$ columns of 
$\partial_p P$ form the earliest basis $B_{opt}$. By reordering the columns
of $\partial_p$, we may assume that $P$ is the identity, giving 
$\partial_p=B_{opt}[I_r\mid R]$.
In this scenario, the columns of $B_{opt}$ form 
a basis of the column-space of $\partial_p$, and
contains the first $r=\rank{(\partial_p)}$ columns of $\partial_p$. 
Note that the $i$-th column in $[I_r\mid \matR]$
gives the coordinate vector of the boundary cycle for $\sigma_i$ 
in the boundary basis $B_{opt}$. 

Take the first $r$ $p$-simplices $\{ \sigma_1, \ldots, \sigma_r\}$.
Their boundaries are linearly independent. Therefore, no subset of them 
can form a $p$-cycle. In analogy to Section \ref{sec:annotation-2}, 
we use $T$ for this collection of $p$-simplices and call them non-sentinel simplices.
The set $\Sigma=\KK_p\setminus T$ of $p$-simplices are the \emph{sentinel simplices}.

Now consider any sentinel $p$-simplex, say $\sigma_{r+i}\in \Sigma$.
Its boundary is the $(r+i)$-th column in $\partial_p$ and is equal to 
$B_{opt}\matR[i]$, where $\matR[i]$ is the $i$-th column of $\matR$. 
This means that $\partial_p \sigma_{r+i} = \sum_{j=1}^r \matR[j,i] (\partial_p \sigma_j)$ where $\matR[j,i]$ is the $j$-th entry in the $i$-th 
column of $\matR$. 
Hence taking the set of $p$-simplices $\sigma_j$, $j \in [1,r]$, 
whose corresponding entries $\matR[j,i]$ are $1$, 
plus $\sigma_{r+i}$ itself, we obtain a $p$-cycle $\cycle(T, \sigma_{r+i})$. 
We call this $p$-cycle a \emph{sentinel cycle}.
Similar to Section \ref{sec:annotation-2}, we set $\cycle(T, \sigma) = 0$ 
for each non-sentinel simplex $\sigma \in T$. 
Clearly, $\gamma(T,\sigma_{r+i})$ can only contain one simplex from $\Sigma$
which is $\sigma_{r+i}$.
We have the desired property:
a sentinel simplex $\sigma\in \Sigma$ belongs
to a sentinel cycle $\gamma(T,\sigma')$ if and only if $\sigma=\sigma'$.
Finally, observe that the columns of 
$\displaystyle{\left[\begin{array}{c}R\\I_{n_p-r}\end{array}\right]}$
give the set of sentinel cycles $Z$. The $(n_p-r) \times (n_p-r)$ 
identity matrix $I_{n_p-r}$ associates
each sentinel cycle $\cycle(T, \sigma)$ in $Z$ to its sentinel $p$-simplex $\sigma$. 
Similar to Proposition \ref{lem:cyclebasis}, we have:
\begin{proposition}
 $Z = \{ \cycle(T, \sigma_{r+1}), \ldots, \cycle(T, \sigma_{n_p}) \}$ 
is a cycle basis, and for any $p$-cycle $z$  we have $z=\sum_{\sigma\in z}\cycle(T, \sigma)$.
\end{proposition}
\begin{proof}
 $Z$ is linearly independent since each cycle contains a unique sentinel simplex. 
Since $Z$ has $n_p-r = n_p - \rank(\partial_p)=\dim(\Z_p)$ elements, $Z$ forms a basis for $\Z_p$. 
An arbitrary $p$-cycle $z$ has a unique coordinate in the cycle basis $Z$. 
Since each sentinel simplex $\sigma \in \Sigma$ belongs to
one and only one cycle in $Z$, the corresponding coordinate is one for a cycle $\cycle(T,\sigma)$ if and only if $\sigma\in z$. 
\end{proof}

Combining this proposition with Step 2 and 3 
from Section \ref{sec:annotation-2}, we obtain the following theorem. 
\begin{theorem}
	We can annotate the $p$-simplices 
	in a simplicial complex with $n$ simplices
	in $O(n^{\omega})$ time.
\label{thm:dimension-p-annotation}
\end{theorem}

\section{Null Homology and Independence}
\label{sec:null-homology}

Our annotation algorithm can be used 
to address some of the computational problems involving $p$-cycles.

\paragraph{Null homology.} A $p$-cycle $z$ in a
simplicial complex $\KK$ is called {\em null homologous}
if $[z]=0$. A cycle is null homologous if and only if it
has zero coordinates in {\em some} and hence {\em any} basis of $\homo{p}(\KK)$.
Consider the problem:
\begin{quote}
Q1: Given a $p$-cycle $z$ in a simplicial complex $\KK$,
decide if $z$ is null homologous.
\end{quote}
A p-cycle $z$ is null homologous if and only if it is linearly
dependent to columns of $\partial_{p+1}$. This could be
determined by checking whether $z$ belongs to
the earliest basis of the matrix $[\partial_{p+1}\mid z]$.
The complexity of such computation is $O(n^\omega)$ (Proposition \ref{lem:linear}(b)).
\remove{
This problem can be solved by matrix reduction 
as follows. Think of inserting an open $(p+1)$-cell whose 
boundary is $z$ into $\KK$ treating it as a cell complex.
If $z$ is already null homologous in $\KK$, the inserted
$(p+1)$-cell creates a $(p+1)$-cycle in the cellular homology of $\KK$.
This can be detected by inserting a column for the
dummy $(p+1)$-cell at the right end of the
boundary matrix $\partial_{p+1}$ of $\KK$. Let $Z$ denote the incidence vector
of $z$, that is, $Z[i]=1$ if and only if $\sigma_i$ is in $z$.
We reduce the matrix $[\partial_p\mid Z]$. 
The cycle $z$ is null homologous if and only if the column for
$Z$ is reduced to a zero column.
This algorithm runs in
$O(n^{\omega})$ time
where $\KK$ contains $n$ simplices. 
}

However, with annotations whose
computations take $O(n^{\omega})$ time, we can improve the query time 
for Q1 to $O(tg)$ where
$g=\dim \homo{p}(\KK)$ and
$t$ is the number of $p$-simplices in $z$. 
For this we simply
add the annotations of the $p$-simplices in $z$
and check if the result is zero, which takes $O(tg)$ time.
Now consider the following decision problem:
\begin{quote}
Q2: Given two $p$-cycles $z_1$ and $z_2$
in a simplicial complex $\KK$,
decide if $z_1$ and $z_2$ are homologous.
\end{quote}
The problem Q2 reduces to Q1 because
$z_1$ and $z_2$ are homologous if and only
if $z_1+z_2$ is null homologous. Therefore,
Q2 can be answered in $O((t_1+t_2)g)$ time
after $O(n^{\omega})$ time preprocessing where
$t_1$ and $t_2$ are the number of $p$-simplices
in $z_1$ and $z_2$ respectively.

\paragraph{Independence.} An analogous problem to testing
null homology is the problem of testing independence.
\begin{quote}
Q3: Find a maximally independent subset
of a given set of $p$-dimensional homology classes
$[z_1],\ldots,[z_k]$ in a simplicial complex $\KK$.
\end{quote}
Without our annotation technique, for each such query, we could apply
Proposition~\ref{lem:linear}(b) to the $n\times (n+k)$ matrix 
$[\partial_{p+1} \mid z_1\,\, z_2\cdots z_k]$.
The set of cycles $z_i$s belonging to the earliest basis would represent
a linearly independent set of classes. 
The overall complexity is $O((n+k)^\omega)$.
Using the iterative technique delineated in Theorem~\ref{sbasis-thm}, 
we can improve the query time to $O(n^\omega(n+k)/n)=O(n^\omega+kn^{\omega-1})$.
However, with our annotation technique, we could improve the
query time. 

Compute the annotations for all edges in $O(n^\omega)$ time. 
By Theorem \ref{thm:dimension-p-annotation}, 
we can then compute an annotation $\ann(z_i)$ in $O(tg)$ time 
for all cycles $z_i$s, where $t$ 
is the number of simplices altogether in all cycles and
each $\ann(z_i)$ is a length-$g$ vector. 
Now construct a matrix $A$ whose $i$th column is the vector
$\ann(z_i)$. Notice that the earliest basis of $A$ form a maximally independent
subset of column vectors in $A$. 
Thus the set of cycles $z_i$s corresponding to columns in this earliest basis 
form a maximally independent subset of the input set of $p$-cycles. 
Since $A$ is of size $g\times k$, we have an
$O(g^\omega+k^\omega)$ query time, after an $O(n^{\omega})$ 
preprocessing time.
We can improve the query time 
to $O((g+k)g^{\omega-1})$ by
using the iterative technique in Theorem~\ref{sbasis-thm}.
Therefore, total time
for computing a maximally independent set takes $O(tg+(g+k)g^{\omega-1})$ time
after annotations.

\section{Conclusions}
\label{sec:conclusions}

In this paper, we present an algorithm to annotate 
$p$-simplices in a complex so that computations 
about the homology groups can be done faster. 
We have shown its applications
to some problems that concern with the optimality of $1$-cycles and 
topological characterizations of the $p$-cycles.
The algorithm uses operations such as matrix inversion 
and matrix multiplication, and thus
can take advantage of the best known
algorithms for these operations, which take $o(n^3)$ time.

For defining the weights of a cycle we used $1$-norm to combine the
weights of the individual edges. For defining the weight of a basis
we also used $1$-norm to combine the weights of the basis cycles.
In these problems we can use any other norm to define these weights.

One may wonder why we cannot extend our annotation approach to 
compute the optimal homology basis or the optimal homologous cycle for higher dimensional cycles. The main 
bottleneck for finding an optimal basis is that the Proposition~\ref{lem:characterization} does not
generalize to higher dimensional cycles. Given that the problems 
are NP-hard in high dimensions even for $g=1$~\cite{cf-hrhl-11}, 
such extensions cannot exist unless P=NP.

Instead of computing the shortest homology basis,
one may want to compute a set of edges with minimal total weight which 
supports a homology cycle basis. In such case, the algorithm in Theorem
\ref{sbasis-thm} cannot be used. The annotations and the covering graph
might help.

Finally, we point out that one can use any finite field instead of $\Z_2$ for
annotations. Since annotations mainly utilize matrix multiplications which 
remain valid under any field, the annotation
algorithm in section~\ref{sec:annotation-p} remains applicable 
without any change. However, the
optimal cycles in sections~\ref{sec:optimality} 
and \ref{sec:localization} require computations of
cycles associated with shortest paths which we do not know how to generalize
for general fields. Specifically, it is not clear how to obtain 
results for applying 
Propositions~\ref{lem:characterization} and \ref{lem:loc-one-comp}. 
This could be a topic of further
research.
\bibliographystyle{abbrv}
\bibliography{surf}

\begin{thebibliography}{10}

\bibitem{bc-tfifm-74}
J.~Bunch and J.~Hopcroft.
\newblock Triangular factorization and inversion by fast matrix multiplication.
\newblock {\em Math.\ Comp.}, 28:231--236, 1974.

\bibitem{ccl-fsntc-10}
S.~Cabello, {\'E}.~Colin~de Verdi{\`e}re, and F.~Lazarus.
\newblock Finding shortest non-trivial cycles in directed graphs on surfaces.
\newblock In {\em Proc. ACM Symp. on Computational Geometry (SOCG)}, pages
  156--165, 2010.

\bibitem{ccl-fctpe-10}
S.~Cabello, {\'E}.~Colin~de Verdi{\`e}re, and F.~Lazarus.
\newblock Finding cycles with topological properties in embedded graphs.
\newblock {\em SIAM J. Disc. Math.}, 25(4):1600--1614, 2011.

\bibitem{cen-mcshc-09}
E.~Chambers, J.~Erickson, and A.~Nayyeri.
\newblock Minimum cuts and shortest homologous cycles.
\newblock In {\em Proc. ACM Symp. on Computational Geometry (SOCG)}, pages
  377--385, 2009.

\bibitem{cf-qhc-07}
C.~Chen and D.~Freedman.
\newblock Quantifying homology classes {II}: Localization and stability.
\newblock {\em CoRR}, abs/0709.2512, 2007.
\newblock Available at \url{http://arxiv.org/abs/0709.2512}.

\bibitem{cf-qhc-08}
C.~Chen and D.~Freedman.
\newblock Quantifying homology classes.
\newblock In {\em Proc. Symp. on Theoretical Aspects of Computer Science
  (STACS)}, pages 169--180, 2008.

\bibitem{cf-hrhl-11}
C.~Chen and D.~Freedman.
\newblock Hardness results for homology localization.
\newblock {\em Discrete and Computational Geometry}, 45(3):425--448, 2011.

\bibitem{ck-11}
C.~Chen and M.~Kerber.
\newblock An output-sensitive algorithm for persistent homology.
\newblock In {\em Proc. ACM Symp. on Computational Geometry (SOCG)}, pages
  207--216, 2011.

\bibitem{cw-mmap-90}
D.~Coppersmith and S.~Winograd.
\newblock Matrix multiplication via arithmetic progressions.
\newblock {\em J. Symb. Comput.}, 9(3):251--280, 1990.

\bibitem{introduction_algorithms}
T.~H. Cormen, C.~E. Leiserson, R.~L. Rivest, and C.~Stein.
\newblock {\em Introduction to Algorithms}.
\newblock {MIT} Press, Cambridge, {MA}, second edition, 2001.

\bibitem{ds-99}
T.~K. Dey and S.~Guha.
\newblock Transforming curves on surfaces.
\newblock {\em Journal of Computer and System Sciences}, 58:297--325, 1999.

\bibitem{DHK10}
T.~K. Dey, A.~Hirani, and B.~Krishnamoorthy.
\newblock Optimal homologous cycles, total unimodularity, and linear
  programming.
\newblock In {\em Proc. ACM Symp. on Theory of Computing (STOC)}, pages
  221--230, 2010.

\bibitem{dsw-alshb-10}
T.~K. Dey, J.~Sun, and Y.~Wang.
\newblock Approximating loops in a shortest homology basis from point data.
\newblock In {\em Proc. ACM Symp. on Computational Geometry (SOCG)}, pages
  166--175, 2010.

\bibitem{e-dgteg-03}
D.~Eppstein.
\newblock Dynamic generators of topologically embedded graphs.
\newblock In {\em Proc. ACM-SIAM Symp. on Discrete Algorithms (SODA)}, pages
  599--608, 2003.

\bibitem{en-mcsns-11}
J.~Erickson and A.~Nayyeri.
\newblock Minimum cuts and shortest non-separating cycles via homology covers.
\newblock In {\em Proc. ACM-SIAM Symp. on Discrete Algorithms (SODA)}, pages
  1166--1176, 2011.

\bibitem{ew-gohhg-05}
J.~Erickson and K.~Whittlesey.
\newblock Greedy optimal homotopy and homology generators.
\newblock In {\em Proc. ACM-SIAM Symp. on Discrete Algorithms (SODA)}, pages
  1038--1046, 2005.

\bibitem{ibarra1982generalization}
O.~Ibarra, S.~Moran, and R.~Hui.
\newblock {A generalization of the fast LUP matrix decomposition algorithm and
  applications}.
\newblock {\em Journal of Algorithms}, 3(1):45--56, 1982.

\bibitem{insw-2011}
G.~F. Italiano, Y.~Nussbaum, P.~Sankowski, and C.~Wulff-Nilsen.
\newblock Improved algorithms for min cut and max flow in undirected planar
  graphs.
\newblock In {\em Proc. ACM Symp. on Theory of Computing (STOC)}, pages
  313--322, 2011.

\bibitem{jeannerod2006lsp}
C.~Jeannerod.
\newblock {LSP matrix decomposition revisited}, 2006.
\newblock Available at
  \url{http://www.ens-lyon.fr/LIP/Pub/Rapports/RR/RR2006/RR2006-28.pdf}.

\bibitem{mms-11}
N.~Milosavljevi{\'c}, D.~Morozov, and P.~{\v S}kraba.
\newblock Zigzag persistent homology in matrix multiplication time.
\newblock In {\em Proc. ACM Symp. on Computational Geometry (SOCG)}, pages
  216--225, 2011.

\bibitem{Munkres}
J.~R. Munkres.
\newblock {\em Elements of Algebraic Topology}.
\newblock Addision-Wesley Publishing Company, Menlo Park, 1984.

\bibitem{p-deeoc-10}
V.~Patel.
\newblock Determining edge expansion and other connectivity measures of graphs
  of bounded genus.
\newblock In {\em ESA 2010}, volume 6346 of {\em Lecture Notes in Computer
  Science}, pages 561--572. Springer, 2010.

\bibitem{stillwell}
J.~Stillwell.
\newblock {\em Classical topology and combinatorial group theory}.
\newblock Springer, 2nd edition, 1993.

\end{thebibliography}

\newpage
\appendix
\label{sec:apx}
\section{Omitted proofs}

\remove{
{\bf Proof of Proposition~\ref{lem:cyclebasis}.}
	Since the edge $e_1\in E\setminus E(T)$ does not appear in 
$\cycle(T,e_2),\dots, \cycle(T,e_k)$, the cycle $\cycle(T,e_1)$ is 
linearly independent of $\cycle(T,e_2),\dots, \cycle(T,e_k)$. 
The same argument applies to any cycle $\cycle(T,e_i)\in Z$,
	and thus the cycles of $Z$ are linearly independent. 
Since $k = \dim(\Z_1)$, $Z$ is indeed a basis for $\Z_1$. 

	To show the second  half of the claim, fix an arbitrary cycle $z$. 
Fix a vertex $s$ of the tree $T$ and for any vertex $u$ of $T$, 
let $T[s,u]$ denote the unique simple path in $T$ connecting $u$ to $s$. 
We then have:
$$
z ~=~ \sum_{e = uv \in z} e ~=~ \sum_{e = uv \in z} (T[s,u] + e + T[s,v]) ~=~ 
		\sum_{e\in z} \cycle(T,e).
$$
The second equality holds as every vertex in the cycle $z$ is the endpoint of an even number of edges of $z$. 
\qed

\medskip
}
{\bf Proof of Proposition~\ref{lem:loc-one-comp}.}
We first argue that there is a bijection between closed walks in $G$ through a 
vertex $v$ with annotation $h$ and walks in $\G$ 
from vertex $(v,0)$ to vertex $(v,h)$.
Indeed, assume first that there is a closed walk 
$w=x_0x_1\cdots x_t$ in $G$ with $x_t=x_0=v$.
By the definition of annotation we have
$$
\ann(w) ~=~ \sum_{i=1,\dots, t} \ann(x_{i-1} x_{i}).
$$
Define $y_0=(x_0,0)$ and define, for each vertex $x_i$ in $w$, the vertex
$y_i=(x_i,\sum_{j\le i}\ann(x_{j-1}x_j))$.
By construction, there is an edge in $\G$ between $y_i$ and $y_{i+1}$,
and therefore $y_0 y_1\cdots y_t$ is a path in $\G$ from $(v,0)$ to $(v,\ann(w))$.
Conversely, for any walk in $\G$ from vertex $(v,0)$ to vertex $(v,h)$,
the projection into the first coordinate provides a closed
walk through $v$ whose annotation is $h$.
	
For all $v\in V$, we compute the shortest paths from $(v,0)$ to all vertices in $\G$, 
and record the lengths.
This is $|V|$ computations of shortest path trees in $\G$ and 
thus takes $O(n (2^g n \log (2^g n))= O(2^g n^2 (g+\log n))$ time
(Erickson and Nayyeri~\cite{en-mcsns-11} can speed up this step because in their case
$\G$ is embedded in a surface.)
The closed walk $w_h$ is then obtained by 
considering the shortest path in $\G$ between $(v,0)$ and $(v,h)$, 
over all vertices $v\in V$, and then taking its
projection onto $G$.
\qed

\medskip
{\bf Proof of Proposition~\ref{prop:fewcycles}.}
We can assume that the homology class is nonzero, 
as otherwise the result is trivial.
Assume for contradiction that the 
shortest cycle $z$ in a given homology class contains more than $g$ 
elementary cycles $z_1,\cdots,z_t$, where $t>g$. 
Each elementary cycle is a cycle by itself. 
Since the homology classes of 
these cycles are not linearly independent, there exists a set
of indices $I\subset \{1,\dots, t\}$ such that $0=\sum_{i\in I} [z_i]$.
Notice that $I\not=\{1,\dots, t\}$ as otherwise 
$[z]=\sum_{i\in \{1,\dots, t\}} [z_i] = 0$. The cycle 
$
	z' = \sum_{i\in \{1,\dots, t\}\setminus I} z_i 
$ 
is strictly shorter than $z$, and represents the same homology 
class as $z$. This contradicts the assumption
that $z$ is shortest in its class.
\qed

\medskip
{\bf Proof of Theorem~\ref{th:OHCP}.}
	We compute annotations for the edges in $O(n^\omega)$ time
	using Theorem~\ref{thm:edgeannotations}.
	For each $h\in (\ZZ_2)^g$ we compute the closed walk $w_h$ given 
	in Proposition~\ref{lem:loc-one-comp}. This takes $O(2^g n^2(g+\log n))$ time.
	  
	For any $h\in (\ZZ_2)^g$ and any integer $k\in [1,g]$, 
	we define $C(h,k)$ as follows
	$$
		C(h,k) = \begin{cases}
			\mbox{length of $w_h$}  &	\mbox{if $k=1$;}\\
			\min \{C(h_1,k-1)+C(h_2,1) \mid h=h_1+h_2\}&	\mbox{if $k>1$.}
		\end{cases}
	$$
	It is straightforward to see by induction
	that $C(h,k)$ is an upper bound 
	on the length of the shortest cycle with annotation $h$. 
	Most interestingly, $C(h,g)$ is the length of the shortest cycle
	with annotation $h$. 
	Indeed, if the shortest cycle with annotation $h$
	consists of the elementary cycles $z_1,\dots, z_t$ with $t\le g$,
	then it follows by induction that, for any $I\subset \{1,\dots t\}$, 
	$C \left( \sum_{i\in I} \ann(z_i), |I| \right)$ is the length of $\left(\sum_{i\in I} z_i \right)$.
	
	A standard dynamic programming algorithm to compute $C(h,k)$ takes $O(2^g)$ time per element,
	for a total of $O(4^gg)$ time. We then return for each $h\in (\ZZ_2)^g$ the cycle
	defining $C(h,g)$. (A shortest cycle homologous to a given cycle $z$
	is obtained as the cycle defining $C(\ann(z),g)$.) 
\qed

\end{document}